\documentclass[a4paper,12pt]{article}
\usepackage[utf8]{inputenc}
\usepackage{amsmath,amsthm,amssymb}
  \usepackage{paralist}
  \usepackage{graphics} 
  \usepackage{epsfig} 
\usepackage{graphicx}  \usepackage{epstopdf}
 \usepackage[colorlinks=true]{hyperref}
\hypersetup{urlcolor=blue, citecolor=red}

\newcommand{\bydef}{:=}
\newcommand{\e}{\mathrm{e}}

\newcommand{\bigO}[1]{\mathop{\mathcal{O}}\left(#1\right)}

\usepackage{xspace}
\newcommand{\A}{\ensuremath{\mathbb{A}}\xspace}

\newcommand{\lo}{\ensuremath{\mathsf{lo}}}

\newcommand{\complete}{\mathrm{C}}
\newtheorem{thm}{Theorem}

\newtheorem{lem}{Lemma}

\title{Asymptotic expression for the fixation probability of a mutant in star graphs}
\author{Fabio A. C. C. Chalub\thanks{Departamento de Matem\'atica and Centro de Matem\'atica e Aplica\c c\~oes, Universidade Nova de Lisboa, Quinta da Torre, 2829-516, Caparica, Portugal. e-mail:chalub@fct.unl.pt, phone: (++351) 212948388}}
\date{\today}

\usepackage{setspace}

\begin{document}

\maketitle

\centerline{Departamento de Matem\'atica and Centro de Matem\'atica e Aplica\c c\~oes,} 
\centerline{Universidade Nova de Lisboa, Quinta da Torre, 2829-516, Caparica, Portugal.}

\begin{abstract}
We consider the Moran process in a graph called the ``star'' and obtain the asymptotic expression for the fixation probability of a single mutant when the size of the graph is large. The expression obtained corrects the previously known expression announced in reference [E Lieberman, C Hauert, and MA Nowak. Evolutionary dynamics on
graphs. Nature, 433(7023):312–316, 2005] and further studied in [M. Broom and J. Rychtar. An analysis of the fixation probability of a
mutant on special classes of non-directed graphs. Proc. R. Soc. A-Math.
Phys. Eng. Sci., 464(2098):2609–2627, 2008]. We also show that the star graph is an accelerator of evolution, if the graph is large enough.
\end{abstract}

\textbf{Keywords}:
Evolutionary graph theory;  Moran process;  Fixation probability;  Asymptotic expansions; Star graph.


\section{Introduction}

The central question in the mathematical study of population genetics is to understand how gene frequencies vary in time. This topic has a long history, since at least the works of Wright and Fisher~\cite{Wright_1931,Fisher_1930}. A simpler model was introduced by Moran~\cite{Moran}: in that case a population of constant size $n$ of two types evolves in discrete time steps. To each type we attribute a positive number, called fitness.  At each time one individual is selected to reproduce with probability proportional to its fitness and one to die (possibly the same one), with probability $1/n$. The so called Moran process continues until the entire population consists of individuals at a given type. We say that one type has \emph{fixated} while the other has become \emph{extinct}.  One particularly important question is the fixation probability of a given type \A, i.e., the probability that after a sufficiently long time, type \A reaches fixation; see~\cite{Traulsenetal2006,Nowak:06} 

The Moran process is a particular case of the so called \emph{Birth-Death (BD) processes with selection on the birth}. Birth-Death process are a particular case of stochastic process in which the number of individuals of a given type varies at most one per time step, i.e., transitions are allowed only between neighboring sites in an appropriate topology. The ``selection at birth'' indicates that the removed individual was selected using an equally distributed random variable, while the random variable indicating selection for reproduction is not necessarily identically distributed among all individuals. See, e.g.,~\cite{Nowak:06,Broom_book} for further informations on BD processes.

In~\cite{Lieberman_Hauert_Nowak_Nature2005}, topology, in the form of a graph, was introduced in the study of population dynamics. Individuals were represented by vertices, while edges represented possible positions for the offspring of individuals in a given vertex. In particular, consider a generic graph where all vertices but one (selected at random) are occupied by a type with fitness 1 (the resident type). The remaining vertex is occupied by a mutant with fitness $r>0$. After a certain fixed time, one of the individuals (resident or mutant) is selected to reproduce with probability proportional to fitness, and its offspring replaces one of the individuals occupying an adjacent vertex, selected with equal probability. This is the so called Moran process on graphs~\cite{Moran,Lieberman_Hauert_Nowak_Nature2005}. The invasion probability $\rho$ is given by the probability that after a sufficiently long time all vertices are occupied by the mutant type (i.e., the mutant has fixed and the resident has become extinct). For a complete graph (all pairs of vertices are connected) with $n$ vertices, the invasion probability is given by $\rho^{\complete}(r,n)\bydef\frac{1-r^{-1}}{1-r^{-n}}$ if $r\ne 1$ and $\rho^{\complete}(1,n)\bydef\frac{1}{n}$. This corresponds to the original problem studied in~\cite{Moran}. We say that a graph is an accelerator of evolution if $\rho>\rho^{\complete}$ if and only if $r>1$. If one of these inequalities is reversed, we say that the graph is a suppressor of evolution.
See~\cite{Nowak:06} for a more detailed explanation of concepts used in this article and also for examples of accelerators and suppressors of evolution.

In this article, we study a specific graph, called the ``star''. 
The star is a graph with $n+1$ vertices, labeled from 0 to $n$, where vertex number 0 is called the \emph{center} and all the others are the \emph{leaves}. The only edges are between vertex 0 and all other vertices. These edges are bidirectional, see figure~\ref{fig}.  In~\cite{Lieberman_Hauert_Nowak_Nature2005}, the star graph was introduced as a particular case of a family of graphs, called the ``superstar'' and parameterized by two natural numbers: $k$, the number of layers and $n$, the number of vertices at layer $k'+1$, for any $k'<k$, connected to any given vertex at layer $k'$. The first layer has a single vertex, called the \emph{center}. The star corresponds to $k=2$. According to the conjecture in~\cite{Lieberman_Hauert_Nowak_Nature2005}, $\rho\approx\frac{1-r^{-k}}{1-r^{-nk}}$ for $n\to\infty$. However, recently~\cite{Diaz} showed that this conjecture is false for $k\ge 5$.

We start this article from a previous work (see~\cite{Broom_Rychtar_PRSA2008}), where an explicit formula for the invasion probability in stars was derived and study the associated asymptotic expression when $n$, the number of individuals, is large. 
The aim of this work is to show that the expression for the leading term for $n\to\infty$ and fixed $r$ in the asymptotic expression of the invasion probability is given by
\begin{equation}\label{estrela_assintotica}
\rho_{\lo}(r,n)\bydef\frac{1-r^{-2}}{1-r^{-2n}\e^{(r^2-1)/r}}.
\end{equation}
A precise definition of $\rho_{\lo}(r,n)$ as the leading term of the exact expression in $n\to\infty$ and fixed $r$  is given in Eq.~(\ref{asymptotic_def}).
\begin{figure}
\centering
 \includegraphics[width=0.4\textwidth]{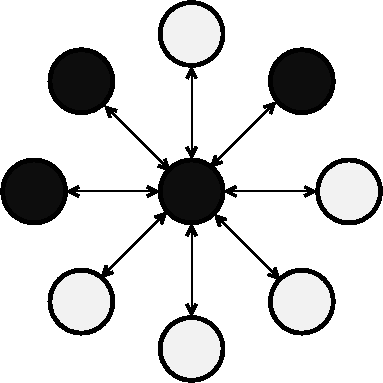}
 \caption{Example of a star graph with $n=8$ (1 center and 8 leaves, totalling 9 vertices). Different colors indicates different occupancies (resident or mutant); links between vertices indicate the edges of the graph.}
 \label{fig}
\end{figure}

Note that this expression is different from the one conjectured in~\cite{Lieberman_Hauert_Nowak_Nature2005}, discussed above, for the case $k=2$.

In section~\ref{sec:evol_star} we discuss the results presented in~\cite{Broom_Rychtar_PRSA2008}. In section~\ref{sec:asymp}, we rigorously derive Eq.~(\ref{estrela_assintotica}) and in section~\ref{sec:complete}, we prove that the star is an accelerator of evolution, if $n$ is large enough. We study this problem numerically and present our conclusions in the final section~\ref{sec:final}.

\section{Evolutionary dynamics in the star graph}
\label{sec:evol_star}

In~\cite{Broom_Rychtar_PRSA2008},
the exact expression
\begin{equation}\label{estrela_prob_fixacao_genr}
 \rho_n(r)=\frac{\frac{r}{n+r}+\frac{n^2r}{nr+1}}{(n+1)\left[1+\frac{n}{n+r}\sum_{j=1}^{n-1}\left(\frac{n+r}{r(nr+1)}\right)^j\right]}\\
\end{equation}
for the invasion probability of mutant with fitness $r$ in a star graph with $n$ leaves was found. 

 We agree with  formula~(\ref{estrela_prob_fixacao_genr}), however the asymptotic expression found in the same reference when $n\to\infty$ seems to not be correct. In particular, the asymptotic expression
 \[
  \frac{n\frac{nr}{nr+1}+\frac{r}{n+r}}{(n+1)\cdot\left(1+\frac{n}{n+r}\sum_{j=1}^{n-1}\left(\frac{n+r}{r(nr+1)}\right)^j\right)}\approx\frac{1}{1+\sum_{j=1}^{n-1}\frac{1}{r^{2j}}}=\frac{1-r^{-2}}{1-r^{-2n}}\ ,
 \]
when $n\to\infty$ is wrong (see~\cite[page 2616]{Broom_Rychtar_PRSA2008}).

\section{Asymptotic Expression}
\label{sec:asymp}

In this section, we will show that 
the associated asymptotic expression for $\rho_n(r)$,  when $n\to\infty$, is given by 
\begin{equation}\label{asymptotic_def}
 \rho_n(r)=\rho_{\lo}(r,n)\left(1+\bigO{n^{-1}}\right)\ ,
\end{equation}
where $\rho(r,n)$ was defined in Eq.~(\ref{estrela_assintotica}).

If $r\ne 1$, Eq.~(\ref{estrela_prob_fixacao_genr}) can be simplified to
\begin{equation}\label{estrela_prob_fixacao}
\rho_n(r)=\frac{r}{n+1}\frac{n^3+n^2r+nr+1}{n^2r+nr^2+n+r}
\left\{1+\frac{n}{n+r}\left[\frac{\frac{n+r}{r(nr+1)}-\left(\frac{n+r}{r(nr+1)}\right)^{n}}{1-\frac{n+r}{r(nr+1)}}\right]\right\}^{-1}\ .
\end{equation}

Now, we consider the asymptotic expression of $\rho_n(r)$ when $n\to\infty$ and $r\ne 1$. Note first that
\[
 \frac{r}{n+1}\frac{n^3+n^2r+nr+1}{n^2r+nr^2+n+r}=1+\bigO{n^{-1}}\ .
 \]
Furthermore
\begin{equation}\label{asymp_inverse}
1+\frac{n}{n+r}\left[\frac{\frac{n+r}{r(nr+1)}-\left(\frac{n+r}{r(nr+1)}\right)^{n}}{1-\frac{n+r}{r(nr+1)}}\right]\\
=1+\frac{1}{r^2-1}\left[1-\left(\frac{n+r}{r(nr+1)}\right)^{n-1}\right]
\end{equation}

We also write
\[
 \frac{n+r}{r(nr+1)}=\frac{1}{r^2}\times \frac{1+\frac{r-1}{n+1}}{1+\frac{1-r}{r(n+1)}}=\frac{1}{r^2}\left(1+\frac{r^2-1}{nr}+\bigO{n^{-2}}\right)\ .
\]
Now, we use that
 \[
 \left(1+\frac{x}{n}+\bigO{n^{-2}}\right)^n=\e^{n\log\left(1+\frac{x}{n}+\bigO{n^{-2}}\right)}=
 \e^{x+\bigO{n^{-1}}}=\e^x\left[1+\bigO{n^{-1}}\right]\ .
  \]
and then
\begin{equation}\label{asymp_frac}
 \left(\frac{n+r}{r(nr+1)}\right)^{n-1}
 =\frac{\e^{(r^2-1)/r}}{r^{2(n-1)}}\left(1+\bigO{n^{-1}}\right)\ .
\end{equation}
We put together Eqs.~(\ref{asymp_inverse}) and~(\ref{asymp_frac}) and find
\begin{align*}
& 1+\frac{n}{n+r}\left[\frac{\frac{n+r}{r(nr+1)}-\left(\frac{n+r}{r(nr+1)}\right)^{n}}{1-\frac{n+r}{r(nr+1)}}\right]\\
&\qquad=1+\frac{1}{r^2-1}-\frac{1}{r^2-1}\frac{\e^{(r^2-1)/r}}{r^{2(n-1)}}\left(1+\bigO{n^{-1}}\right)\\
& \qquad
 =\frac{1}{1-r^{-2}}-\frac{1}{1-r^{-2}}\frac{\e^{(r^2-1)/r}}{r^{2n}}\left(1+\bigO{n^{-1}}\right)\\
 &\qquad=
 \frac{1-r^{-2n}\e^{(r^2-1)/r}}{1-r^{-2}}\left(1+\bigO{n^{-1}}\right)\ .
\end{align*}
Gathering all asymptotic expansions, and using the definition~(\ref{asymptotic_def}), we find expression~(\ref{estrela_assintotica}).

\section{The star is an accelerator of evolution}
\label{sec:complete}

Now, let us compare the Moran process in the star and in the complete graph. 

For neutral evolution ($r=1$), we have the following simple result:
\begin{lem}
$\rho_n(1)=\frac{1}{n+1}=\rho^\complete(1,n+1)$.
\end{lem}

\begin{proof}
See Eq.~(\ref{estrela_prob_fixacao_genr}) with $r=1$. Equivalently, this result can be proved from Eq.~(\ref{estrela_prob_fixacao}) in the limit $r\to 1$.
\end{proof}

A similar result is valid for a large class of graphs; in particular, it is valid for graphs with at most one \emph{root} (a root is a vertex with no edge leading to it). However, a single mutant cannot invade a graph with two or more roots, and, in this case, the invasion probability is equal to zero.  See~\cite{Nowak:06}.

Note that
$\lim_{r\to 1}\rho(r,n)=\frac{1}{n-1}\ne\frac{1}{n+1}$, showing that the large population limit $n\to\infty$ and the weak selection limit $r\to 1$ are not interchangeable.

For $r\ne1$, we show that the star is an accelerator of evolution:
\begin{thm}\label{accelerator}
Consider $r>1$ ($r<1$). Then, for $n$ large enough,  
$\rho_n(r)>\rho^{\complete}(r,n+1)$ ($<$, respect.), and therefore the star is an accelerator of evolution. 
\end{thm}

\begin{proof}
We initially see that
\[
 \lim_{n\to\infty} \frac{\rho_{\lo}(r,n)}{\rho^{\complete}(r^2,n+1)}=\lim_{n\to\infty} \frac{1-r^{-2(n+1)}}{1-r^{-2n}\e^{(r^2-1)/r}}=\left\{
 \begin{array}{ll}
 r^{-2}\e^{-(r^2-1)/r}>1\ ,&\quad r<1\ ,\\
 1\ ,&\quad r>1\ . 
\end{array}\right.
\]
On the other hand
\begin{equation*}
 \frac{\rho^{\complete}(r^2,n+1)}{\rho^{\complete}(r,n+1)}=\frac{(1-r^{-(n+1)})(1-r^{-2})}{(1-r^{-2(n+1)})(1-r^{-1})}\ .
\end{equation*}
Taking the limit $n\to\infty$, we find
\[
 \lim_{n\to\infty} \frac{\rho^{\complete}(r^2,n+1)}{\rho^{\complete}(r,n+1)}=\left\{
 \begin{array}{ll}
  0\ ,&\quad r< 1\ ,\\
  \frac{1-r^{-2}}{1-r^{-1}}=\frac{r+1}{r}\ ,&\quad r>1\ .
 \end{array}\right.
\]
Finally, we write  
\begin{equation}\label{limit_n}
 \frac{\rho_n(r)}{\rho^{\complete}(r,n+1)}=\left(1+\bigO{n^{-1}}\right)\times\frac{\rho_{\lo}(r,n)}{\rho^{\complete}(r^2,n+1)}\times\frac{\rho^{\complete}(r^2,n+1)}{\rho^{\complete}(r,n+1)}
\end{equation}
and taking the limit $n\to\infty$, we conclude that
\[
\lim_{n\to\infty} \frac{\rho_n(r)}{\rho^{\complete}(r,n+1)}=\left\{
\begin{array}{ll}
 0\ ,&\quad r<1\ ,\\
 \frac{r+1}{r}>1\ ,&\quad r>1\ .
\end{array}\right.
\]
For $n$ large enough $\rho_n(r)>\rho^{\complete}(r,n+1)$ if and only if $r>1$.
\end{proof}

\section{Discussion}
\label{sec:final}

Despite the fact that the expression conjectured in~\cite{Lieberman_Hauert_Nowak_Nature2005} is wrong for $k=2$ and $k\ge 5$, it has been widely used (e.g.,~\cite{Frean_Rainey_Traulsen,Zhang_etal_2012,Broom_etal_2012}). However, exact or asymptotic expressions for $k>2$ have not been found and should be the subject of further investigations. See also~\cite{Allen_Nowak,Shakarian} for a  review of evolutionary graph theory.

Differences between expression~(\ref{estrela_assintotica}) and the expression $\frac{1-r^{-2}}{1-r^{-2n}}$ given by~\cite{Broom_Rychtar_PRSA2008} and~\cite{Lieberman_Hauert_Nowak_Nature2005} are hardly noticeable from the numerical point of view. 
In fact, for $r>1$, the difference between both expressions are exponentially small, as $r^{-n}\to 0$ exponentially fast when $n\to\infty$. However, for $r<1$ these differences, despite being small, cannot be neglected. 
See table~\ref{tab} and figure~\ref{plot} for numerical comparisons.

\begin{center}
\begin{table}
\centering
 \begin{tabular}{cc|c|cc|cc}
  $r$&$n$&$\rho_n(r)$&$\rho_{\lo}(r,n)$&$\epsilon$&$\frac{1-r^{-2}}{1-r^{-2n}}$&$\epsilon$\\
  \hline
$0.800$&$500$&$1.08\times 10^{-97}$&$1.09\times 10^{-97}$&$0.00629$&$6.92\times 10^{-98}$&$0.358$\\
$0.800$&$1000$&$1.33\times 10^{-194}$&$1.34\times 10^{-194}$&$0.00320$&$8.51\times 10^{-195}$&$0.360$\\
$0.900$&$500$&$5.04\times 10^{-47}$&$5.06\times 10^{-47}$&$0.00509$&$4.10\times 10^{-47}$&$0.186$\\
$0.900$&$1000$&$8.83\times 10^{-93}$&$8.85\times 10^{-93}$&$0.00255$&$7.17\times 10^{-93}$&$0.188$\\
  $0.950$& $1000$& $3.34\times 10^{-46}$& $3.35\times 10^{-46}$&$0.00224$&$3.03\times10^{-46}$&$0.0955$\\
  $0.950$&$5000$&$2.06\times 10^{-224}$&$2.06\times 10^{-224}$&$0.000448$&$1.86\times 10^{-224}$&$0.0971$\\
  $0.990$&$1000$&$3.85\times 10^{-11}$&$3.86\times 10^{-11}$&$0.00203$&$3.78\times 10^{-11}$&$0.0179$\\
  $0.990$&$5000$&$4.66\times 10^{-46}$&$4.66\times 10^{-46}$&$0.000459$&$4.57\times 10^{-46}$&$0.0195$\\
  $1.01$&$1000$&$0.0197$&$0.0197$&$0.00204$&$0.0197$&$0.00204$\\
  $1.05$&$1000$&$0.0928$&$0.0930$&$0.00197$&$0.0930$&$0.00197$\\
  $1.01$&$5000$&$0.0197$&$0.0197$&$0.000484$&$0.0197$&$0.000484$\\
  $1.05$&$5000$&$0.0929$&$0.0930$&$0.000410$&$0.0930$&$0.000410$\\
 \end{tabular}

\caption{Invasion probabilities for different values of $r$ and $n\gg 1$: $\rho_n$ indicates the exact value given by Eq.~(\ref{estrela_prob_fixacao_genr}); in the third column, we compute the
approximate value given by Eq.~(\ref{estrela_assintotica}) and its associated error $\epsilon=|\rho_{\lo}(r,n)-\rho_n(r)|/\rho_n(r)$; the last column indicates the approximation in references~\cite{Lieberman_Hauert_Nowak_Nature2005,Nowak:06,Broom_Rychtar_PRSA2008} and its associated error. Note that for $r>1$, the two approximations are essentially equivalent; for $r<1$ the first approximation is clearly superior. All calculations were performed in Sage v. 5.13.}
\label{tab}
\end{table}
\end{center}

\begin{figure}
 \centering
 \includegraphics[width=0.45\textwidth]{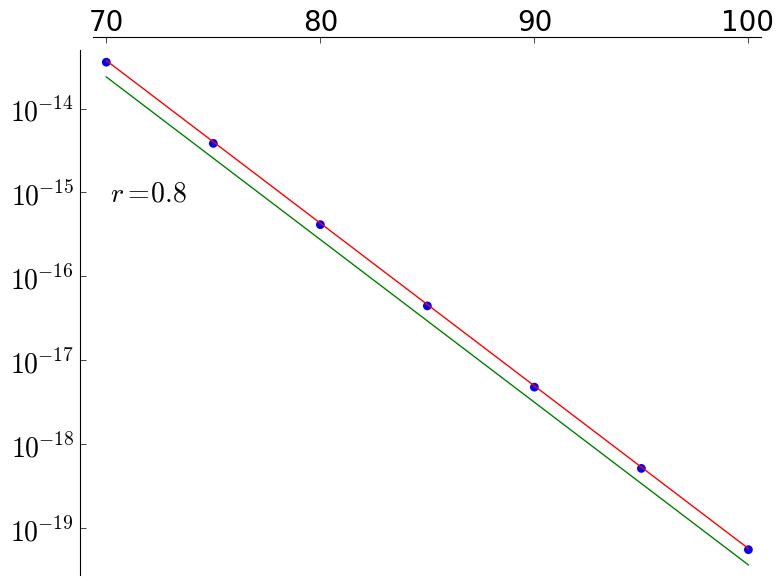}\quad
 \includegraphics[width=0.45\textwidth]{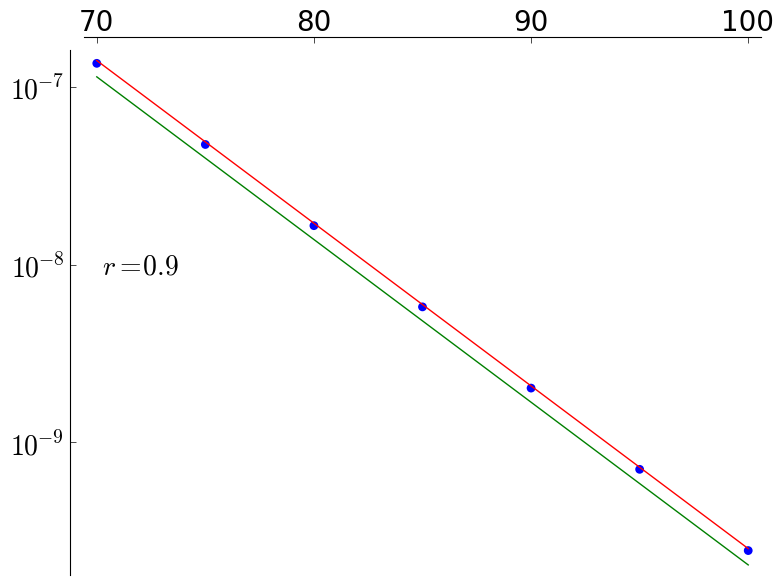}\newline
 \includegraphics[width=0.45\textwidth]{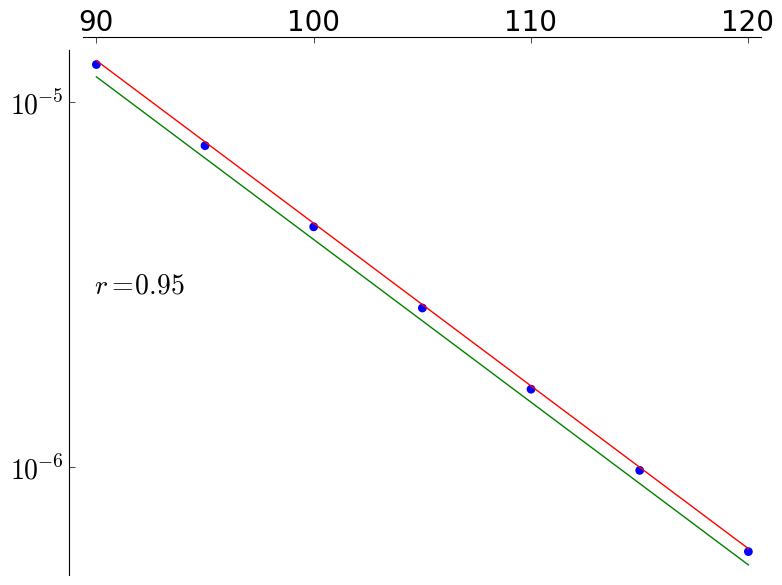}\quad
 \includegraphics[width=0.45\textwidth]{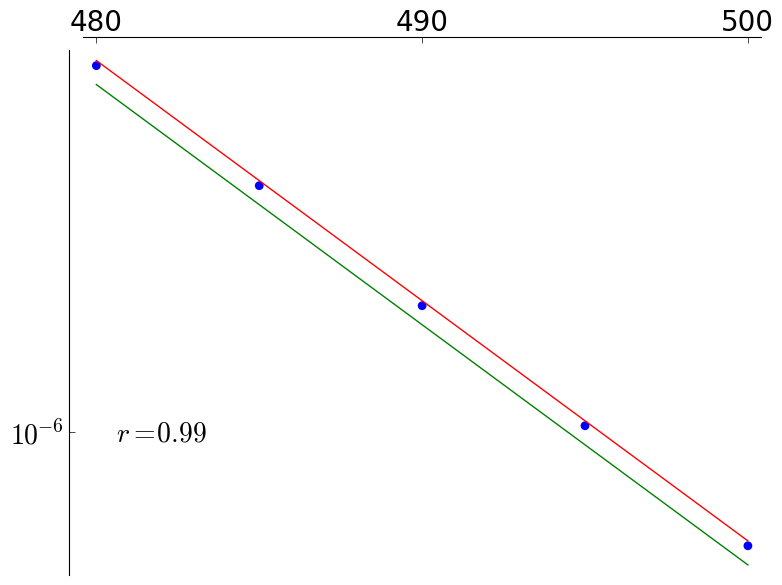}
\caption{(Color online) For different values of $n$ ($x$ axis), we plot $\rho_n(r)$ (blue dots), $\rho_\lo(n,r)$ (red line) and $\frac{1-r^{-2}}{1-r^{-2n}}$ (green line). We use a logarithmic scale in the $y$ axis. Note that the red line is consistently a better approximation than the green one with respect to the exact values.}
\label{plot}
 \end{figure}

 \section*{Acknowledgements}
 This work was supported by CMA/FCT/UNL and FCT/Portugal, under the project  UID/MAT/00297/2013. FACCC has also a ``Investigador FCT'' (FCT/Portugal) grant.
 FACCC also thanks the careful reading of the manuscript by the referees, in particular an inconsistency pointed out by one of the reviewers that allowed an improvement in the statement and proof of Theorem~\ref{accelerator}.
%

\end{document}